\documentclass[11pt]{article}
\usepackage[margin=1.1in]{geometry}
\usepackage{amsmath,amssymb,amsthm,mathtools}
\usepackage{graphicx,booktabs,microtype}
\usepackage[round]{natbib}
\usepackage[colorlinks=true,allcolors=blue]{hyperref}
\newtheorem{theorem}{Theorem}
\newtheorem{proposition}[theorem]{Proposition}
\newtheorem{corollary}[theorem]{Corollary}
\newcommand{\E}{\mathbb{E}}
\newcommand{\R}{\mathbb{R}}
\newcommand{\norm}[1]{\left\lVert#1\right\rVert}

\newcommand{\proj}{\Pi_\Omega}
\newcommand{\Ht}{\mathcal H_t}
\title{\textbf{Geometric comparisons of electoral rules under feedback}}
\author{Sumit Mukherjee\\[0.3em]\small Amazon, USA}
\date{}

\begin{document}
\maketitle

\begin{abstract}
We study how the policy selected in one election changes the
electorate that votes in the next. In our spatial feedback model, we
measure policy coverage by the distance from the implemented policy
to the farthest voter and disagreement by the variance of voter
positions. We bound voter disagreement and candidate dispersion after
an election and extend the voter bound to a finite sequence of
elections. To illustrate why we measure travel by camp, we give an
exact example in which policies leave the same disagreement and move
different groups of voters. In paired simulations,
policies chosen to reduce next-round disagreement can leave lower
terminal variance and a wider gap between the principal camps than
policies chosen for coverage. Average travel can also be lower across
the electorate and higher within a camp.
When candidates adjust their platforms, paths that elect the same
winners can implement different policies.
\end{abstract}


\section{Introduction}
\label{sec:intro}

For a fixed electorate and candidate slate, we can compare voting
rules by the policies they select. In repeated elections we can also
compare the trajectories they generate. Voters may respond to the
implemented policy, and candidates may adjust to changes in support.
The next election can therefore take place with different voter
positions and candidate platforms.

While simulating repeated elections with the spatial voting model
of \citet{mukherjee2026electoral}, we asked whether the policy that
minimizes the distance to the farthest voter also minimizes
disagreement after voters respond. We measure disagreement by the
variance of voter positions. Two paths can end with the same variance
even if one moves both camps and the other moves mostly one camp.
We therefore also track cumulative movement per voter in each camp.

We make two contributions. First, we give conditional upper bounds
on disagreement among voters and dispersion among candidates after
an election. The voter bound uses the \emph{winner radius} $R_t$,
the largest distance from the implemented policy to a voter. For a
specified voter response we also give a bound that uses the list of
voter distances to the policy. The candidate bound uses the
\emph{supporter-centroid radius} $S_t$, the largest distance from a
platform to its supporter centroid, and a condition on how supporter
centroids vary across platforms. When the one-round voter factors
have a common upper bound below one, we obtain a bound on expected
disagreement over several elections. In a noiseless example with two
point-mass camps, selecting either endpoint of the segment between
the camps leaves the same disagreement. The choice of endpoint
determines which camp moves.

Second, we use a common policy grid to compare the choices of two
objectives at a given electorate and the paths they produce over
several elections. At the same voter positions, we compare the policy
that minimizes the winner radius with the policy that minimizes
predicted next-round disagreement. We then let each objective choose
the policy for the electorate it produced. In the unequal-camp runs
under Sorting pressure, which combines social attraction with
backlash, disagreement minimization sometimes left lower terminal
variance and a wider gap between the principal camps than coverage
did. We use travel per voter to see which camps made the adjustment.
We also compare Score and Plurality on a common slate and use a
separate Plurality comparison to examine what changes when platforms
move.

Section~\ref{sec:preliminaries} defines the quantities we compare.
Section~\ref{sec:model} gives the model and the mathematical
results, and Section~\ref{sec:sim} presents the simulations.

\subsection{Related Work}

\paragraph{Feedback in repeated elections.}
Models of repeated choice differ in what changes between decisions.
\citet{ParkesProcaccia2013} treat preference profiles as states and
allow a decision to affect the transition to the next state. A
spatial model can also let platforms change. One of the party
strategies in \citet{laver2005} moves the platform toward its
supporters' preferred positions. \citet{BorgersEtAl2025} study
voter and candidate positions that evolve together in continuous
time. In our model, voters respond to the selected policy, and
moving candidates use support computed before the update. We hold
Plurality fixed to compare policies implemented with a fixed slate
and with moving platforms.

\paragraph{Voter disagreement and group separation.}
The same overall spread can accompany different degrees of
separation between groups. \citet{EstebanRay1994} distinguish
identification within groups from alienation across them. In the
Bridge conflict experiments, we measure the gap between the two
principal camp centers and each camp's distance to the bridge
group center. The weighted squares of those distances give the
between-group component. We use the three distances and overall
variance to see whether the principal camps draw closer under the
same objective that leaves lower variance across the electorate.
We also record cumulative travel per voter in each group to track
how far its voters move along the two paths.

\section{Preliminaries}
\label{sec:preliminaries}

\subsection{Spatial elections}
There are $n\ge1$ voters with positions $x_1,\ldots,x_n$ and
$K\ge1$ candidates with platforms $c_1,\ldots,c_K$ in a compact convex
policy space $\Omega\subset\R^p$, where $p\ge1$. A voter prefers a
platform closer to their position in Euclidean distance. A voting
rule uses the resulting ballots to choose an implemented policy
$w\in\Omega$. Under Plurality, for example, each voter votes for
the nearest platform and the candidate with the most votes
implements its platform. In the oracle simulations we also
consider policies chosen directly from a finite grid.

We index rounds by $t$ and select the policy for a round using the
state at its start. When analyzing a single round, we omit $t$ and use
primes for the positions after the update.

\subsection{Coverage, disagreement and movement}
For an implemented policy $w$, write $d_i=\norm{x_i-w}$. For each
candidate let $s_j$ denote the center of its support and let
$e_j=\norm{c_j-s_j}$. Section~\ref{sec:feedback} gives the support
model. Define
\begin{equation}
 R=\max_i d_i,\quad S=\max_j e_j,\qquad
 D=\frac1n\sum_i\norm{x_i-\bar x}^2,\quad
 P=\frac1K\sum_j\norm{c_j-\bar c}^2 .
 \label{eq:measures}
\end{equation}
Here $\bar x=n^{-1}\sum_i x_i$ and $\bar c=K^{-1}\sum_j c_j$.
On a fixed electorate, the policy with the smaller winner radius $R$
gives better coverage. The supporter-centroid radius $S$ measures the
largest gap between a platform and its support center. Voter
disagreement $D$ and candidate dispersion $P$ measure spread within
the electorate and slate, respectively. The quantity $D$ is the trace
of the empirical covariance of voter positions with normalization $1/n$.
To relate dispersion to distances between voters, we use the identity
\begin{equation}
 \frac{1}{n}\sum_i\norm{x_i-\bar x}^2
   =\frac{1}{2n^2}\sum_{i,j}\norm{x_i-x_j}^2.
 \label{eq:pairs}
\end{equation}
To separate the within-group and between-group contributions to $D$,
fix a partition of the electorate into nonempty groups $G_g$. Let
$\pi_g=|G_g|/n$ be the share of group $g$, let $m_g$ be its mean, and
let $V_g$ be the mean squared distance to that mean. Then
\begin{equation}
 D=\sum_g\pi_g V_g
       +\sum_{g<h}\pi_g\pi_h\norm{m_g-m_h}^2.
 \label{eq:decomp}
\end{equation}
The first term gives within-group dispersion and the second gives
between-group dispersion. To derive them, write
$x_i-\bar x=(x_i-m_g)+(m_g-\bar x)$ and sum the squared distances over
each group. The cross term vanishes because the deviations from the
group mean sum to zero. Applying the weighted pairwise identity to
the group means gives the second term.

We measure cumulative movement per voter in each group as
\begin{equation}
 B_g(T)=\frac{1}{|G_g|}\sum_{i\in G_g}\sum_{t=0}^{T-1}
             \norm{x_{i,t+1}-x_{i,t}} .
 \label{eq:travel}
\end{equation}
The labels are fixed. A group's voters can travel far even if its
centroid moves little, since paths can change direction and the
individual movements can cancel in the mean.

\section{Model and Analysis}
\label{sec:model}

We model how voter positions and candidate platforms change after an
election. Voters respond to the selected policy, and candidates can
adjust their platforms using support from that round. We derive
dispersion bounds from these updates and use a two-camp example to
compare policies chosen for coverage and next-round disagreement.

\subsection{Voter and candidate updates}
\label{sec:feedback}
Write $\proj$ for Euclidean projection onto $\Omega$.
After the policy is implemented, voter $i$ moves according to
\begin{equation}
x_i'=\proj\{x_i+g(d_i)(w-x_i)+\varepsilon_i\}.
\label{eq:voter}
\end{equation}
Here $\varepsilon_i$ is the voter shock. We use an attraction
function $g$ with $0<\eta_{\min}\le g(d)\le\eta_{\max}<1$ on
$[0,\operatorname{diam}(\Omega)]$. A constant $g$ moves every voter
the same fraction of the way toward the policy. A distance-dependent
$g$ allows that fraction to vary. Without shocks the updated position
is a convex combination of $x_i$ and $w$ and already lies in $\Omega$.

To update a candidate's platform, we specify a target $s_j\in\Omega$
using the electorate and slate before the update. We can obtain this
target as a weighted mean of voter positions
\begin{equation}
s_j=\frac{\sum_i a_{ij}x_i}{\sum_i a_{ij}},
\qquad a_{ij}\ge0,\quad \sum_{j=1}^{K}a_{ij}=1.
\label{eq:support}
\end{equation}
We use the electorate mean if the denominator is zero. When the
election selects a candidate, the ballots determine the winner.
The support weights determine the targets used to update platforms.
The candidate update is
\begin{equation}
 c_j'=\proj\{c_j+h(e_j)(s_j-c_j)+\mu(\bar x-c_j)+\delta_j\}.
 \label{eq:candidate}
\end{equation}
Here $\delta_j$ is the candidate shock. The coefficient $h(e_j)$
governs the step toward the candidate's supporters and $\mu\ge0$
the step toward the electorate mean. With a fixed slate we hold all $c_j$ constant. When
candidates move, both updates use the same state at the start of the
round.

\subsection{Policy distance and voter disagreement}
\label{sec:drift}
To bound disagreement after one election, we use pairwise distances
between the updated voter positions in \eqref{eq:pairs}. Let $\Ht$
be the information generated by
the state and implemented policy before the shocks. We assume
$\E[\varepsilon_i\mid\Ht]=0$ and
$n^{-1}\sum_i\E[\norm{\varepsilon_i}^2\mid\Ht]\le \sigma_x^2$.
The shocks may be correlated.

Let $z_i$ be the unprojected position before the shock. Suppose
a nonnegative factor $q$ satisfies
$\norm{z_i-z_j}\le q\norm{x_i-x_j}$ for every pair.
Projection is nonexpansive \citep{BauschkeCombettes2017}.
Consequently the dispersion after projection is at most the
dispersion of $z_i+\varepsilon_i$. Expanding that dispersion around
the sample mean makes the cross term vanish in conditional expectation,
and the dispersion of the shocks is at most their mean squared norm.
Equation~\eqref{eq:pairs} now gives
\begin{equation}
 \E[D'\mid\Ht]\le q^2D+\sigma_x^2.
 \label{eq:drift}
\end{equation}
The same argument applies to the candidate update. It remains to find
pairwise factors that we can compute from the current state.

\begin{theorem}[Winner radius bound]
\label{thm:radius}
For the voter update in \eqref{eq:voter}, suppose $g$ has
Lipschitz constant $L_g$ on its distance domain. Under the
voter shock conditions above,
\[
 \E[D'\mid\Ht]\le (1-\eta_{\min}+L_gR)^2D+\sigma_x^2 .
\]
\end{theorem}

\begin{proof}
With $z_i=x_i+g(d_i)(w-x_i)$, write
\[
 z_i-z_j=(1-g(d_i))(x_i-x_j)+(g(d_j)-g(d_i))(x_j-w).
\]
The first coefficient is at most $1-\eta_{\min}$. The reverse triangle
inequality gives $|d_i-d_j|\le\norm{x_i-x_j}$ and the final vector has norm
at most $R$. Taking norms gives the bound by \eqref{eq:drift}.
\end{proof}

For a constant response $g\equiv\eta$ and zero shocks, the policy
term cancels in every difference $x_i'-x_j'$. The pairwise factor is
$1-\eta$, and the variance after $T$ rounds is $(1-\eta)^{2T}D_0$
for any sequence of policies. More generally, a bound of the form
\eqref{eq:drift} with $q<1$ gives a decrease in conditional mean
when $D>\sigma_x^2/(1-q^2)$. With zero shocks this condition is
simply $D>0$.

\subsubsection{A bound using the observed voter distances}
The winner-radius bound uses a range and a Lipschitz constant for the
response. When the response is known at the occupied distances, we can
compute a bound from those values directly. This calculation also
allows voters to respond to a common anchor. Suppose the voter response
has one term for the policy $w$ and another for a common anchor $m$,
such as the electorate mean. The update before projection is
\begin{equation}
 z_i+\varepsilon_i,\qquad
 z_i=x_i+a(m-x_i)+\gamma(d_i)(w-x_i).
 \label{eq:anchored}
\end{equation}
The quantities $a$, $m$ and the function $\gamma$ are fixed given
the current state and policy. Positive coefficients cause attraction
and negative coefficients cause repulsion. The winner-only response
uses $a=0$ and $\gamma=g$. The \emph{Sorting pressure} response uses
$m=\bar x$ and $a>0$. Its policy coefficient $\gamma$ is positive at
small distances and negative at sufficiently large distances.
Appendix~\ref{app:bridge-settings} gives the coefficients used in the
simulations. For $d_i>0$, write $x_i=w+d_i u_i$ with
$\norm{u_i}=1$. Then
\begin{equation}
 z_i=w+a(m-w)+f(d_i)u_i,\qquad f(d)=d[1-a-\gamma(d)].
 \label{eq:radial}
\end{equation}
At $d_i=0$ the last term vanishes. The first two terms are common
to all voters, so they cancel when we compare a pair.

\begin{theorem}[Bound from the occupied distances]
\label{thm:occupied}
Under the conditional shock assumptions above, the variance bound
\eqref{eq:drift} for \eqref{eq:anchored} holds with
\begin{equation}
 q_{\rm occ}=\max\left\{
 \max_{k:\,d_{(k)}>0}\frac{|f(d_{(k)})|}{d_{(k)}},\
 \max_{1\le k<m_d}\frac{|f(d_{(k+1)})-f(d_{(k)})|}{d_{(k+1)}-d_{(k)}}\right\},
 \label{eq:qocc}
\end{equation}
where $d_{(1)}<\cdots<d_{(m_d)}$ are the distinct occupied distances
and empty maxima are zero. With $w$, $a$, the anchor and $\gamma$ held
fixed, $q_{\rm occ}$ is the exact Lipschitz constant of the
unprojected map on the union of the full spheres of these radii in
$\R^p$.
\end{theorem}

\begin{proof}
For two positive radii $r,s$, let $\theta$ be the angle between
the directions from $w$. Squared separation before and after the
unprojected update is
\[
 (r-s)^2+2rs(1-\cos\theta),\qquad
 [f(r)-f(s)]^2+2f(r)f(s)(1-\cos\theta).
\]
For distinct radii, any secant is a weighted average of adjacent
secants. Equal radii give a zero first term. Also
$f(r)f(s)\le |f(r)f(s)|\le q_{\rm occ}^2rs$, including when one
residual changes sign. For a zero radius the angular term vanishes.
Thus each updated squared distance is at most $q_{\rm occ}^2$ times
its original value. Applying \eqref{eq:drift} gives the variance
bound. For sharpness on full spheres, points on a common ray realize
the absolute adjacent secants, and antipodal points realize
$|f(d)|/d$ at a positive radius. With all radii zero the domain is
a singleton.
\end{proof}

The computation sorts $n$ distances and makes one pass through the
distinct values. It uses no angles. On a given point cloud, the
observed directions and projection can yield a smaller dispersion.
For example, with $w=0$ and voters at $0$ and $1$, take $g(d)=0.1$ up to
$d=0.9$ and let it increase linearly to $0.9$ at $d=1$. The
winner-radius bound gives $8.9$. The two occupied radii give
$q_{\rm occ}=0.1$, which is the actual factor for that pair.

For voters at the same positive radius, the angular term in
\eqref{eq:qocc} bounds the factor by which the unprojected update
without shocks changes their separation. To see why the secants also
matter, take $w=0$ and two voters at $0.2$ and $0.3$ with the Sorting
pressure response of Section~\ref{sec:bridge}. The anchor is their mean
$m=0.25$ and the shocks are zero. At both distances $f(d)/d<1$,
but $f(d)=d(0.955+0.125d)$ on the interval between them.
The secant is $1.0175$, and the pairwise gap grows by that factor.
Projection leaves the new positions $0.201$ and $0.30275$ unchanged.
Both residual radii shrink relative to the translated center, yet
their difference grows.

\subsubsection{A bound over several rounds}
The preceding bounds use the electorate and policy at the start of
each round. We now use a common upper bound on the one-round factors
to bound disagreement over a sequence of elections. Write $D_t$ for
disagreement after $t$ updates and $q_t$ for a pairwise factor in
\eqref{eq:drift} computed at the state and policy before the round-$t$
shocks.

\begin{corollary}[Finite-horizon voter disagreement]
\label{cor:horizon}
Suppose the voter shock conditions for \eqref{eq:drift} hold at each
round $t=0,\ldots,T-1$, where $T\ge1$, and $q_t\le q_\star$
almost surely for a deterministic $0\le q_\star<1$. Then
\begin{equation}
 \E[D_T]\le q_\star^{2T}\E[D_0]
     +\sigma_x^2\frac{1-q_\star^{2T}}{1-q_\star^2}.
 \label{eq:horizon}
\end{equation}
\end{corollary}

\begin{proof}
Since $D_t\ge0$, taking expectations in \eqref{eq:drift} gives
$\E[D_{t+1}]\le q_\star^2\E[D_t]+\sigma_x^2$.
Iterating the inequality gives \eqref{eq:horizon}.
\end{proof}

With zero shocks, we can calculate a bound using the factors at the
electorates we observe along a path. Applying the pairwise inequality
and nonexpansive projection at each update gives
$D_T\le D_0\prod_{t=0}^{T-1}q_t^2$.
In Section~\ref{sec:oracles} we use both the common upper bound and
these pathwise products to evaluate the noiseless oracle runs.

\subsection{Supporter distance and candidate dispersion}
\label{sec:candidatebound}
Each candidate's update includes a step toward a target computed
from its own support. To bound the distance between two updated
platforms, we relate the
distance between their targets to the distance between the platforms
before the update.

\begin{proposition}[Supporter-centroid radius bound]
\label{prop:candidate}
Suppose $\norm{s_j-s_\ell}\le L_s\norm{c_j-c_\ell}$ for every pair of
current targets and platforms. Let $h$ have Lipschitz constant $L_h$
and values in $[\lambda_{\min},\lambda_{\max}]\subseteq[0,1]$. Define
\begin{align}
 \kappa&=\max_{\lambda\in\{\lambda_{\min},\lambda_{\max}\}}
                  |1-\mu-\lambda|,\\
 q_c&=\kappa+\lambda_{\max}L_s+L_h(1+L_s)S .
 \label{eq:qc}
\end{align}
If $\E[\delta_j\mid\Ht]=0$ and
$K^{-1}\sum_j\E[\norm{\delta_j}^2\mid\Ht]\le\sigma_c^2$, then
\[
 \E[P'\mid\Ht]\le q_c^2P+\sigma_c^2 .
\]
\end{proposition}

\begin{proof}
Put $\lambda_j=h(e_j)$. The difference of two unprojected
noiseless candidate updates is
\[
 (1-\mu-\lambda_j)(c_j-c_\ell)+\lambda_j(s_j-s_\ell)
        +(\lambda_\ell-\lambda_j)(c_\ell-s_\ell).
\]
The alignment condition and the reverse triangle inequality give
$|\lambda_\ell-\lambda_j|\le L_h(1+L_s)\norm{c_j-c_\ell}$.
The remaining gap has norm at most $S$. Applying
\eqref{eq:pairs} and the projection argument gives the stated bound.
\end{proof}

The supporter-centroid radius enters the bound through differences in
the step sizes of candidates with different supporter gaps. Its
contribution vanishes for constant $h$. Two nearby platforms can have
widely separated targets under nearest-candidate assignment, so the
alignment condition must be checked for the support map in use. For the Gaussian weights used below,
$a_{ij}\propto\exp(-\norm{x_i-c_j}^2/\sigma^2)$ within each voter,
the alignment condition holds with
$L_s=\operatorname{diam}(\Omega)^2/(2\sigma^2)$
(Appendix~\ref{app:alignment}).

\subsection{An exact comparison in two camps}
\label{sec:tradeoff}

We use two point-mass camps to compare policies on a fixed electorate.
Put the camps at $-1$ and $1$ with shares $\pi$ and $1-\pi$, where
$0<\pi<1$. There are no shocks and the policy lies in $[-1,1]$.
The midpoint minimizes the winner radius $R(w)=1+|w|$. We compare
it with the policies that minimize next-round disagreement and use
those policies to see how a choice that leaves the same disagreement
can change who moves.

Write $\varphi(d)=d g(d)$ for the distance a voter moves at initial
distance $d$. A policy $w$ moves the left camp by
$b_L(w)=\varphi(1+w)$ and the right camp by $b_R(w)=\varphi(1-w)$.
The left camp moves right and the right camp moves left. Their updated
gap remains positive because $b_L+b_R\le 2\eta_{\max}<2$.
The next-round disagreement is therefore
\begin{equation}
 D_+(w)=\pi(1-\pi)[2-b_L(w)-b_R(w)]^2 .
 \label{eq:twocamp}
\end{equation}

\begin{proposition}[Disagreement and movement in two camps]
\label{prop:twocamp}
In the noiseless two-camp setting above, suppose $\varphi$ is
continuous and strictly convex on $[0,2]$. Then the
two endpoints minimize $D_+$ and the midpoint maximizes it.
For each $d$ with $D_+(1)\le d<D_+(0)$, exactly two policies
have $D_+(w)=d$. The policies are $w$ and $-w$, with $w>0$,
and they interchange the two camps' movements.
\end{proposition}

\begin{proof}
The sum $\varphi(1+w)+\varphi(1-w)$ is strictly convex and even, hence
strictly increasing on $[0,1]$. The positive gap in
\eqref{eq:twocamp} is then strictly decreasing on that interval.
The claim follows from continuity and the symmetry of the gap. The
movements are interchanged by replacing $w$ with $-w$.
\end{proof}

For example, take $g(d)=1/5+d/10$. At the midpoint both camps move
$3/10$ and the remaining gap is $7/5$. At either endpoint one camp
moves $4/5$ and the other stays in place, leaving a gap of $6/5$.
The two endpoint policies give the same variance for any
camp sizes, while the population-weighted movement can differ because
the moving camp can have either size. A fair lottery between the two
reflected policies equalizes expected per-capita movement and gives
the same realized disagreement whichever policy is selected. With a
constant response, every policy on the line gives the same next-round
variance.

Under the conditions of Proposition~\ref{prop:twocamp}, restrict
policies at each round to the interval between the current camps.
Choosing an endpoint in every round then minimizes terminal
disagreement. The choice of endpoint determines which camp moves.
Appendix~\ref{app:repeated-camps} gives the argument.

\subsubsection{When changes within and between camps differ}
When voters are dispersed within camps, overall disagreement $D$
can fall as the camp centers move farther apart. For
example, use $w=0$ on $[0,1]$ and $g(d)=(1+d)/4$.
Give the two camps positions $\{12/25,12/25\}$ and $\{0,1\}$.
The camp means move from $0.48$ and $0.50$ to $0.3024$ and $0.25$.
Thus the centroid gap grows from $0.02$ to $0.0524$ while $D$
falls from $0.1251$ to about $0.03194$. The positions of these two
labeled camps are interleaved. In simulations we report
\eqref{eq:decomp} when a comparison is about the separation of camps.

\section{Simulations}
\label{sec:sim}

The two-camp calculation compares policies on a fixed electorate.
We use the same comparison in our simulations to measure what each
objective gains when both policies act on the same voters. We also
start both trajectories from the same electorate to see how their
separate choices shape later electorates. The Bridge experiment
asks how the voter response affects which groups travel farther.
In the ballot experiments, we ask whether the trade-off between
coverage and disagreement also arises when rules select from a
candidate slate, and whether the winning candidate identifies the
policy along a sequence of elections.

\subsection{Experimental design and policy objectives}
\label{sec:simdesign}
On the unit square we draw two Gaussian camps centered at
$(0.26,0.56)$ and $(0.74,0.44)$. Each coordinate has standard deviation
$0.05$, independently of the other. We clip the initial positions to
the square and keep the initial camp labels throughout. In this
baseline a voter at distance $d$ from the policy moves by the fraction
$g(d)=0.05+0.17d/\sqrt2$ toward the policy, with no shocks.

We initialize 24 electorates with 72 voters on the left and 168 on
the right. The \emph{coverage oracle} selects the policy that
minimizes the maximum voter distance. The \emph{disagreement oracle}
selects the policy that minimizes voter variance after the next
update. Both compare all 289 policies in $\{0,1/16,\ldots,1\}^2$.
We break computed ties by taking the first policy in lexicographic
order. The two runs in a pair share the initial electorate, voter
response, policy grid, and 12-round horizon. Neither objective uses
candidate positions.

At a fixed electorate $x$, let $R_x(w)$ be the winner radius and
$D_+(x;w)$ the voter variance after one update with policy $w$. We
find both oracle policies, denoted $w_R(x)$ and $w_D(x)$, at every
visited electorate. The reduction in winner radius is
$R_x(w_D)-R_x(w_R)$ and the reduction in next-round variance is
$D_+(x;w_R)-D_+(x;w_D)$. Both are nonnegative by construction.
Evaluating both policies on the same electorate shows how much
the choice of objective changes the current radius and the next
voter update. Along each path, the oracle acts on the electorate
produced by its earlier choices. We record the radius at each
decision, disagreement after each update, and movement using
\eqref{eq:travel}. We summarize outcomes across initializations
with medians and interquartile ranges. For paired comparisons,
we take differences within each pair before summarizing them.

Appendix~\ref{app:grid} bounds how much the grid can raise the
minimum radius and next-round standard deviation at a fixed
electorate. Minimizing next-round standard deviation selects the
same policies as minimizing next-round variance.

\subsection{Better coverage can leave voters more dispersed}
\label{sec:oracles}
The coverage policy had a smaller winner radius than the
disagreement policy on all 24 common starting electorates. After
one update on those same electorates, the disagreement policy left
lower voter variance. Table~\ref{tab:oracle-levels} reports the
medians for the two objectives, and Appendix~\ref{app:oracle-checks}
compares their choices on later electorates.

Figure~\ref{fig:oracles} follows each objective as it begins making
decisions on the electorate produced by its own earlier policies.
After 12 updates, median voter variance was much lower under
disagreement minimization than under coverage. The disagreement
oracle selected a policy on the boundary at every decision and had
a larger winner radius than coverage at every matched decision.

For the response in this comparison, write $b=0.17/\sqrt2$. The
residual in \eqref{eq:radial} is $f(d)=0.95d-bd^2$. At a positive
radius $d$, the angular factor is $0.95-bd$. At distinct radii
$r<s$, the absolute secant is $0.95-b(r+s)$, which is positive
on the feasible distance interval and at most the angular factor
at $s$. Hence $q_{\rm occ}=0.95-bd_{\min}^+$, where $d_{\min}^+$
is the smallest positive occupied distance. When every distance
is zero, $q_{\rm occ}=0$.

The zero-shock bound following Corollary~\ref{cor:horizon} gives
$D_{12}\le0.95^{24}D_0$ in each draw. Using the diameter in the
winner-radius factor gives a factor above one. The occupied
distances along each recorded path give tighter bounds than the
uniform occupied-distance bound. Table~\ref{tab:oracle-levels}
reports these pathwise bounds alongside the observed
terminal-to-initial variance ratios. Both were smaller under
disagreement minimization than under coverage.

In every pair, cumulative travel per voter was higher in both camps
under disagreement minimization than under coverage. The smaller
camp traveled farther per voter than the larger one in only one of
the 24 disagreement paths. The two-camp example shows why the camp
farther from a policy can move more in a single round. The
farthest camp can change over a sequence of elections.

\begin{figure}[tb]
\centering
\includegraphics[width=.96\linewidth]{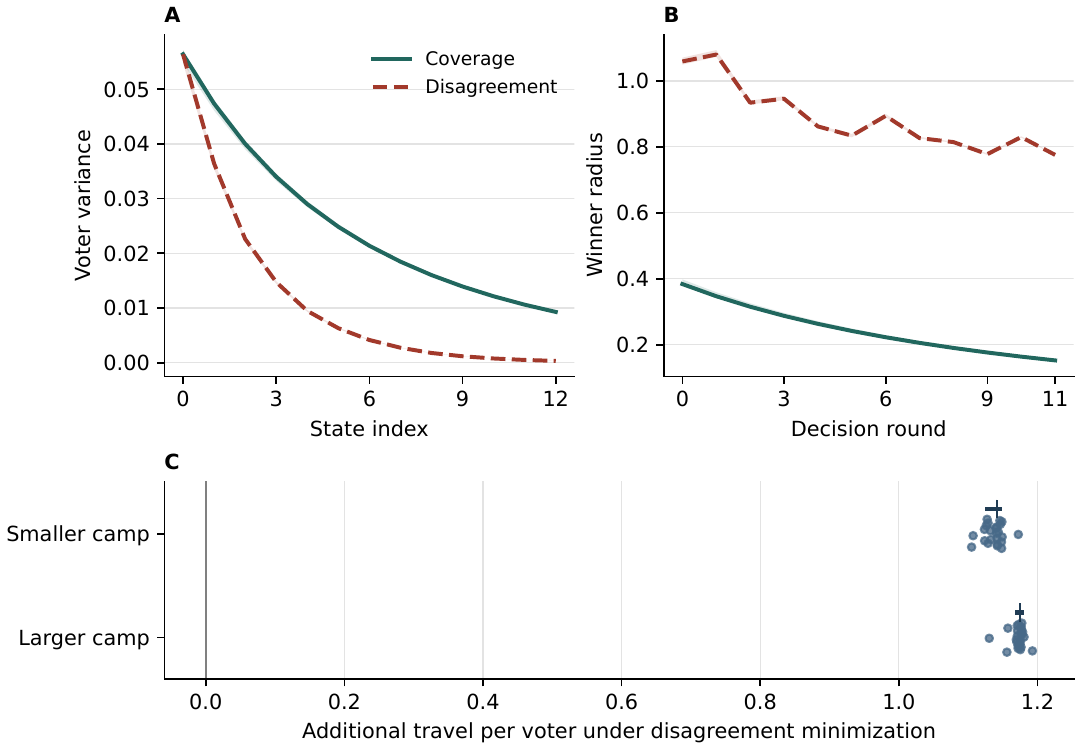}
\caption{Voter variance (A), winner radius (B), and additional cumulative
travel per voter (C) under the two oracle objectives. Curves give
medians and bands span interquartile ranges over the 24 paired
initializations. State zero is the common initialization and state
12 follows the last decision, in round 11. In (C), each point is
a within-seed difference, disagreement minus coverage. Vertical
offsets separate the points, and the bars show the interquartile
ranges and medians. Radius and travel use policy-space units, and
variance uses their square.}
\label{fig:oracles}
\end{figure}

\subsection{Average travel can hide how far the smaller camp moves}
\label{sec:bridge}
In every winner-only pair in Bridge conflict, cumulative travel per
voter was higher under disagreement minimization than under coverage
both across the electorate and within each group. With social
attraction and backlash, the comparison depended on whose travel we
averaged. In 12
of 24 unequal-camp pairs on the primary grid, the disagreement path
had lower average travel across the electorate and higher average
travel in the smaller camp than the coverage path.

The Bridge electorate comes from the simulation notebooks. It has
two principal camps and a third group between them. Of 500 voters,
160 belong to the bridge group. The left and right camps contain
238 and 102 voters in the $70{:}30$ design and 170 each in the
$50{:}50$ design.

We use 24 paired draws per split and make 16 updates on the same
$17$ by $17$ grid. We compare the winner-only response
\eqref{eq:voter} with the notebook's Sorting pressure response in
\eqref{eq:anchored}. The primary winner-only runs have no shocks.
In the primary Sorting runs, the two paths within a pair share
Gaussian shocks. The disagreement oracle forecasts next-round
variance with projection and zero shocks. The Sorting paths add
their shared shocks before projecting the updated voter positions.
Appendix~\ref{app:sim} gives the full responses and initialization.

\begin{table}[htbp]
\centering
\caption{Disagreement minus coverage in the Bridge conflict electorate.
Entries are medians of within-seed differences over 24 pairs. Square
brackets give interquartile ranges. Travel gives each voter equal
weight. Splits give the sizes of the left and right principal camps,
in that order.}
\label{tab:bridge}
\begin{tabular}{lccc}
\toprule
Response & Split & Terminal variance $\times 10^3$ & Cumulative travel\\
\midrule
Winner-only & $70{:}30$ & $-6.245$ $[-6.384,-6.151]$ & $1.5731$ $[1.5646,1.5746]$\\
Winner-only & $50{:}50$ & $-6.767$ $[-6.881,-6.651]$ & $1.5598$ $[1.5552,1.5639]$\\
Sorting pressure & $70{:}30$ & $-1.771$ $[-3.712,-0.983]$ & $-0.00161$ $[-0.00360,-0.00041]$\\
Sorting pressure & $50{:}50$ & $-2.490$ $[-2.983,-1.273]$ & $-0.00137$ $[-0.00197,-0.00072]$\\
\bottomrule
\end{tabular}
\end{table}

On the primary grid, disagreement minimization left lower terminal
variance than coverage in every pair under both responses
(Table~\ref{tab:bridge}). To check whether the travel comparison
depended on shocks, we ran both responses with and without them.
For both camp splits, the median paired comparison showed higher
average travel under disagreement minimization than under coverage
with winner-only attraction. With Sorting pressure, the same
comparison showed slightly lower average travel. Both directions
held with and without shocks (Appendix~\ref{app:shock}).
On a finer grid, individual Sorting pairs can change their travel
ordering even when the ordering of terminal variance persists
(Appendix~\ref{app:oracle-checks}).

\subsection{Lower overall disagreement can coexist with a wider gap between camps}
\label{sec:separation}
In the unequal-camp Sorting runs on the primary grid, disagreement
minimization left lower terminal voter variance than coverage in all
24 pairs. In 21 of those pairs, the centers of the two principal camps
ended farther apart than under coverage. Relative to the common start,
however, the centers usually ended closer together under both
objectives. The larger terminal gap under disagreement minimization
usually meant that the camps converged less.

With shared shocks, the camp centers ended closer together under
disagreement minimization than under coverage in every balanced
pair on all three grids. In the unequal-camp design they ended
farther apart under disagreement minimization in most pairs on
each grid (Table~\ref{tab:bridge-gaps}).

The between-group component includes distances from each camp
to the bridge. For each pair of groups we computed the terminal
distance $\norm{m_g-m_h}$ and its weighted squared contribution
$\pi_g\pi_h\norm{m_g-m_h}^2$ in \eqref{eq:decomp}. In the unequal-camp
runs on the primary grid, the left--bridge distance was smaller
under disagreement minimization than under coverage in all 24 pairs,
while the bridge--right distance was larger in 23. The two bridge
distances can therefore make opposing contributions to the change in
between-group dispersion. In 10 of the 21 pairs with a larger gap
between the principal camps, the weighted sum over the three pairs
was lower under disagreement minimization. Between-group dispersion
can miss a difference in how far apart the principal camps end up.

We also compared each path's terminal variance with its variance
at the common start on the primary grid. Under Sorting pressure,
nearly all coverage paths in each camp split finished above their
initial variance. About half the disagreement paths did so as well.
A path can thus leave voters less dispersed than under coverage
and still leave them more dispersed than they were initially.
Figure~\ref{fig:bridge} shows how variance changed over the rounds.

\begin{figure}[tb]
\centering
\includegraphics[width=.94\linewidth]{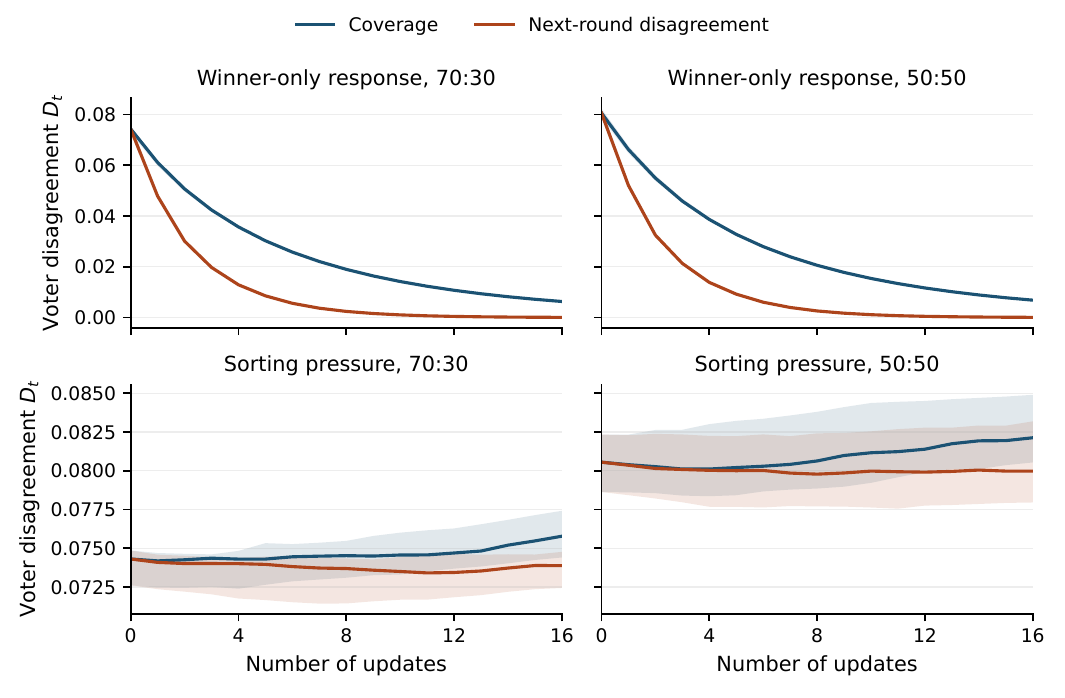}
\caption{Voter variance in the Bridge conflict comparison. The columns
show the two splits of the principal camps. Shading spans the
interquartile range across the 24 paired initializations and curves are
pointwise medians. The rows have different vertical scales. State zero
is common to both arms in each pair.}
\label{fig:bridge}
\end{figure}

\subsection{Score can improve coverage and leave more disagreement after one round}
The trade-off also arose when the rules chose from the same
candidate slate. In 63 of 64 common starting electorates, Score
chose a policy with a smaller winner radius than Plurality and
left higher voter variance after the first update. Both rules
selected the same platform in the remaining draw.

We used the six platforms in Appendix~\ref{app:sim}, two camps of
72 and 168 voters, and the noiseless winner-only response. A
voter at distance $d$ assigns the candidate a score $1-d/\sqrt2$.
The candidate with the highest mean score wins. Score selected
the same platform in every draw. When the rules differed, Plurality
always selected the same alternative platform
(Table~\ref{tab:ballot-detail}).

\subsection{The same winners can implement different policies}
\label{sec:plurality}
In 63 of 64 noiseless unequal-camp Plurality pairs, the fixed-slate
run and the run with Gaussian support weights elected the same
candidate round by round. In those pairs, the two paths implemented
different policies in every round after the first.

To distinguish the effect of re-election from the effect of moving
platforms, we compare three paths from the same first election.
One keeps the first policy fixed. In the other two we repeat
Plurality on the original slate or on a slate whose candidates
update their platforms.

We use 64 initializations and 16 rounds for each of the $72{:}168$
and $120{:}120$ camp splits with the initial slate in
Appendix~\ref{app:sim} and the baseline voter response. When
candidates move, we use \eqref{eq:candidate} with Gaussian support
weights of bandwidth $1.4$, $h(e)=0.12+0.12\min(e,1)$, and
$\mu=0.25$. Both updates use the state before the election, and
candidates have no shocks. We repeat the comparison with shared
voter shocks. In a fourth path we use nearest-candidate support
with the same movement coefficients.

\begin{table}[htbp]
\centering
\caption{Paired terminal voter variance contrasts for Plurality over 64
initializations per row. Splits give left and right camp sizes,
in that order. Variances are multiplied by $10^3$. The Gaussian comparison changes platform movement with repeated elections
in both arms. Re-election in the first two rows uses a fixed slate.
Brackets give interquartile ranges. The last column counts positive
paired variance differences.}
\label{tab:ballot}
\begin{tabular}{lclrr}
\toprule
Split & Voter shock SD & Comparison & Median [IQR] & Positive\\
\midrule
$30{:}70$ & $0$ & Re-elect minus first policy & $0\ [0,0]$ & $0/64$\\
$50{:}50$ & $0$ & Re-elect minus first policy & $0.762\ [0.679,1.581]$ & $64/64$\\
$30{:}70$ & $0$ & Gaussian motion minus fixed slate & $2.203\ [2.169,2.225]$ & $63/64$\\
$50{:}50$ & $0$ & Gaussian motion minus fixed slate & $0.873\ [0.779,0.970]$ & $64/64$\\
$50{:}50$ & $0.01$ & Gaussian motion minus fixed slate & $0.300\ [0.104,0.839]$ & $64/64$\\
\bottomrule
\end{tabular}
\end{table}

The fixed-slate comparison shows what changes when we keep
holding elections. Without voter shocks, re-election in the
$30{:}70$ design kept selecting the first policy and reproduced
the fixed-policy paths. In the balanced design, re-election left
higher terminal voter variance than holding the first policy in
all 64 draws (Table~\ref{tab:ballot}). The between-camp component
was higher in every draw, while the median paired change in
pooled within-camp dispersion was a decrease.

In the noiseless unequal-camp comparison, the platforms updated
using Gaussian support became much less dispersed over time.
Voter variance after the last update was higher than under
repeated elections on the fixed slate in 63 of 64 draws. The
median paired difference showed lower cumulative travel per voter
in each camp with moving platforms
(Table~\ref{tab:ballot-detail}).

The way candidates responded to support mattered in the balanced
runs with common shocks. Relative to repeated elections on the
fixed slate, the median paired change in terminal variance was an
increase with Gaussian support and a decrease with nearest support
(Tables~\ref{tab:ballot} and \ref{tab:ballot-detail}). Both assignments
use the same synchronous update and Plurality ballot.

\section{Conclusion}
\label{sec:discussion}

The winner radius enters our bound on expected voter disagreement
after one election. When we know the response at the distances
voters actually occupy, we can also calculate a bound from those
distances. A common upper bound below one on the one-round factors
gives a bound on expected disagreement over several elections. The
candidate bound uses the distances from platforms to their supporter
centroids and accounts for how those centroids vary across candidates.

The paired simulations show why we measure both coverage and
disagreement. At a common electorate, the policy with better
coverage can leave voters more dispersed after they respond. In the
unequal-camp runs with social attraction and backlash, disagreement
minimization sometimes left lower terminal variance and less
convergence between the principal camps than coverage did. Relative
to coverage, travel across the electorate could also be lower while
voters in the smaller camp traveled farther per voter. Overall
variance tells us how dispersed voters are across the electorate.
The change in the camp gap measures how much the principal camps
converge, and travel per voter shows how far voters in each camp move.

In repeated Plurality, paths that elect the same candidates round
by round can implement different policies when platforms move.
Recording both the policy and the winner helps distinguish changes
in where candidates stand from changes in who wins. Our comparisons
describe movement among modeled ideological positions. Applying
them to a particular electorate would require evidence on how its
voters and candidates respond between elections.

\bibliographystyle{plainnat}
\bibliography{polarization_refs}

\appendix
\section{Simulation settings and checks}
\label{app:sim}
The oracle initializations use seeds $260700,\ldots,260723$.
The Plurality initializations use seeds $261000,\ldots,261063$, with
separately seeded streams for the two camps. The initial slate, in
index order, is
\[
 (0.16,0.67),\ (0.24,0.45),\ (0.43,0.56),\
 (0.57,0.44),\ (0.76,0.63),\ (0.84,0.33).
\]
In Plurality, we use the lowest index to resolve ties in distance and
in the number of first preferences. We use the same convention for
exact ties in mean score under Score. Both rules implement the selected
candidate's position. We use the electoral simulator of
\citet{mukherjee2026electoral} to construct and aggregate
ballots and to draw the initial Gaussian camps in the two-camp oracle
experiment. In the ballot experiments we draw clipped Gaussians
directly from their stated parameters. The complete configurations and
numerical records are available with the
\href{https://github.com/mukhes3/eletoral_dynamics/tree/code-v44}{replication materials}.

\subsection{Bridge conflict settings}
\label{app:bridge-settings}
We use seeds $2603400,\ldots,2603423$. The left, bridge, and right
components have means $(0.22,0.58)$, $(0.52,0.52)$, and $(0.80,0.42)$
and covariance matrices
\[
 \begin{pmatrix}.010&.002\\.002&.011\end{pmatrix},\quad
 \begin{pmatrix}.008&.001\\.001&.009\end{pmatrix},\quad
 \begin{pmatrix}.010&-.002\\-.002&.011\end{pmatrix},
\]
respectively. We reuse the group-specific streams across balance
settings and clip initial draws to the square. In Sorting pressure
the update before projection is
\[
 x_i+0.02(\bar x-x_i)
  +\left[0.01-0.11\min\!\left\{1,\max\!\left\{0,
       \frac{d_i-0.12}{0.88}\right\}\right\}\right](w-x_i)
  +\varepsilon_i .
\]
The two arms use the same shocks by voter and round, with independent
normal coordinates of mean zero and standard deviation $0.004$.
In the Plurality shock sensitivity we use independent coordinate
shocks with standard deviation $0.01$, again shared across the paths
within each draw. The nearest-support rule assigns each voter to the
closest platform, breaking ties by index. It uses the same zero-support
fallback and candidate response as the Gaussian variant.

\subsection{Crossing response and voter shocks}
\label{app:shock}
To compare each response at both shock settings, we repeat the 24
paired draws for each camp split with zero shocks and with shared
Gaussian shocks of coordinate standard deviation $0.004$. In both
settings, the policy forecasts include projection and use zero shocks. Table~\ref{tab:shock}
gives the two additional cells. In the balanced noiseless Sorting
design, disagreement minimization leaves lower terminal variance
and less average travel than coverage in 22 of 24 pairs. The two
objectives give the same values in the remaining pairs. With
unequal camps, terminal variance is lower under disagreement
minimization in all 24 pairs and average travel is lower in 17.

\begin{table}[htbp]
\centering
\caption{Median paired differences, disagreement minus coverage, in
the additional shock-control cells for the Bridge conflict design.
Splits give left and right principal camp sizes, in that order.}
\label{tab:shock}
\begin{tabular}{lccc}
\toprule
Response & Split & Terminal variance $\times10^3$ & Cumulative travel\\
\midrule
Winner-only, shocks & $70{:}30$ & $-6.392$ & $1.55237$\\
Winner-only, shocks & $50{:}50$ & $-6.895$ & $1.54085$\\
Sorting pressure, no shocks & $70{:}30$ & $-1.900$ & $-0.00216$\\
Sorting pressure, no shocks & $50{:}50$ & $-1.963$ & $-0.00203$\\
\bottomrule
\end{tabular}
\end{table}

\subsection{Same-state comparisons and grid sensitivity}
\label{app:oracle-checks}
Table~\ref{tab:oracle-levels} reports the outcomes and pathwise bounds
from the baseline runs. To compare the two policies at the last
decision, we evaluate both on each of the electorates reached by
the two paths. The median reduction in next-round variance from
choosing the disagreement policy is smaller on the electorates
reached under disagreement minimization than on those reached
under coverage (Table~\ref{tab:oracle-check-values}).

\begin{table}[!htbp]
\centering
\caption{Baseline oracle results over 24 paired starts. The first
three rows give medians within each arm. The next two rows give
ranges across paths. The last two rows give median paired
differences, disagreement minus coverage, in cumulative travel
per voter over 12 updates. Radius and travel use policy-space
units and variance uses their square.}
\label{tab:oracle-levels}
\begin{tabular}{lrr}
\toprule
Quantity & Coverage & Disagreement\\
\midrule
Winner radius at the initial state & $0.3838$ & $1.0589$\\
Voter variance after one update & $0.04736$ & $0.03641$\\
Voter variance after 12 updates & $0.009289$ & $0.000337$\\
Pathwise upper bound on $D_{12}/D_0$ & $[0.218,0.252]$ & $[0.0376,0.0400]$\\
Observed $D_{12}/D_0$ & $[0.159,0.171]$ & $[0.00584,0.00609]$\\
\midrule
\multicolumn{2}{l}{Cumulative travel per voter} & Median paired difference\\
Smaller camp & & $1.141$\\
Larger camp & & $1.174$\\
\bottomrule
\end{tabular}
\end{table}

\begin{table}[!htbp]
\centering
\caption{Median paired differences in the baseline same-state
comparisons and the Bridge Sorting grid check. The baseline
radius reduction is $R_x(w_D)-R_x(w_R)$ and the next-round variance reduction
is $D_+(x;w_R)-D_+(x;w_D)$. The terminal Bridge contrast is
disagreement minus coverage. Radius is in policy-space units
and variance in their square. Each row uses 24 starts.}
\label{tab:oracle-check-values}
\begin{tabular}{lrr}
\toprule
Baseline electorate used for both choices & Radius reduction & Variance reduction\\
\midrule
Common initial state & $0.68153$ & $0.010941$\\
Last decision on the coverage path & $0.70779$ & $0.002445$\\
Last decision on the disagreement path & $0.72672$ & $0.000124$\\
\midrule
\multicolumn{2}{l}{Sorting pressure on the $33\times33$ grid} & Terminal variance contrast\\
$70{:}30$ split & & $-0.002744$\\
$50{:}50$ split & & $-0.001554$\\
\bottomrule
\end{tabular}
\end{table}

To see whether the comparisons depend on the policy grid under
Sorting pressure, we keep the starts and shocks fixed. We repeat
all 24 pairs per camp split on a $33$ by $33$ grid and the first
eight pairs per split on a $65$ by $65$ grid. Terminal variance
is lower under disagreement minimization in every pair on all
tested grids. In the unequal-camp design, five pairs have higher
average travel under disagreement minimization on the primary grid.
Average travel in all five is lower under disagreement minimization
on the $33$ by $33$ grid. The same ordering holds in every pair
tested on the two finer grids at both splits. Under winner-only
attraction, the first-four-seed refinement preserves the ordering
of terminal variance and travel at both camp splits.

The separation results use the terminal centroids of the fixed
groups. Table~\ref{tab:bridge-gaps} reports the full comparisons
with shared shocks and the matched subset for unequal camps. In
the unequal-camp primary-grid runs with shared shocks, the
principal-camp gap grew from its initial value on only one
disagreement path and on no coverage path. In the noiseless
primary-grid Sorting control, the terminal gap was larger under
disagreement minimization than under coverage in 19 of 24
unequal-camp pairs and none of the balanced pairs.

\begin{table}[htbp]
\centering
\caption{Terminal group separation under Sorting pressure with
shared voter shocks. The gap is the distance between the left and
right camp centres, in policy-space units. Differences are
disagreement minus coverage within each pair. The matched rows
use the same first eight starts across grids. The splits refer to
the principal camps. Both designs include a separate bridge group.}
\label{tab:bridge-gaps}
\begin{tabular}{llrrrr}
\toprule
Policy grid & Starts & \shortstack{Median gap\\difference} & Larger gap &
Smaller gap & \shortstack{Higher between-\\group dispersion}\\
\midrule
\multicolumn{6}{l}{\emph{$70{:}30$ principal-camp split}}\\
$17\times17$ & All 24 & $0.00820$ & $21/24$ & $3/24$ & $11/24$\\
$33\times33$ & All 24 & $0.00273$ & $17/24$ & $7/24$ & $0/24$\\
$17\times17$ & Matched 8 & $0.00715$ & $7/8$ & $1/8$ & $4/8$\\
$33\times33$ & Matched 8 & $0.00194$ & $6/8$ & $2/8$ & $0/8$\\
$65\times65$ & Matched 8 & $0.00516$ & $6/8$ & $2/8$ & $1/8$\\
\midrule
\multicolumn{6}{l}{\emph{$50{:}50$ principal-camp split}}\\
$17\times17$ & All 24 & $-0.00785$ & $0/24$ & $24/24$ & $0/24$\\
$33\times33$ & All 24 & $-0.00519$ & $0/24$ & $24/24$ & $0/24$\\
$65\times65$ & First 8 & $-0.00464$ & $0/8$ & $8/8$ & $0/8$\\
\bottomrule
\end{tabular}
\end{table}

As a tie check on the original grid, we choose the last coverage
policy in lexicographic order within 32 double-precision machine
epsilons of the minimum radius under Sorting pressure. This
changes one coverage path in each split and preserves the signs
of the terminal variance contrasts. With constant response
$g\equiv0.10$ and zero shocks, the observed variances in both
arms agree with $0.9^{2t}D_0$ to within $3\times10^{-17}$. The
disagreement objective uses the lexicographically first policy
among its ties. Even with a constant response, the chosen policy
can affect cumulative movement.

\subsection{Projection at fixed states and policies}
\label{app:projection}
We apply the occupied-distance bound at each of 1,536 recorded
state--policy pairs per response. Its factor is below one for all
winner-only pairs. Under Sorting pressure the secant term exceeds
one in every pair, leaving contraction undecided by the bound.
We evaluate each selected Sorting policy again at its original
pre-update state with shocks set to zero. Before projection,
139 of the 1,536 forecasts contract. Clipping to the square
changes an expansion into a contraction in another 607 cases.
Each comparison holds the state and policy fixed.

\subsection{Additional ballot results}
Table~\ref{tab:ballot-detail} gives exact summaries for the
same-slate Score comparison and the repeated Plurality experiments.

\begin{table}[!htbp]
\centering
\caption{Additional ballot outcomes. Each experiment uses 64
starts per setting. Signed contrasts subtract the second condition
named in each panel from the first, within each draw. All
variance values are unscaled. Candidate dispersion is the
variance of platform positions. Travel is in policy-space units
per voter.}
\label{tab:ballot-detail}
\begin{tabular}{lr}
\toprule
Outcome & Value\\
\midrule
\multicolumn{2}{l}{\emph{Score versus Plurality on the same initial slate}}\\
Platform selected by Score in every draw & $(0.57,0.44)$\\
Platform selected by Plurality in 63 of 64 draws & $(0.84,0.33)$\\
Median paired winner radius contrast & $-0.2910$\\
Median paired next-round voter variance contrast & $0.005080$\\
\midrule
\multicolumn{2}{l}{\emph{Balanced noiseless Plurality, re-election versus first policy}}\\
Median paired terminal pooled within-camp dispersion contrast & $-0.000027$\\
\midrule
\multicolumn{2}{l}{\emph{Unequal noiseless Plurality, Gaussian support versus fixed slate}}\\
Initial candidate dispersion & $0.07659$\\
Median terminal candidate dispersion with Gaussian support & $2.19\times10^{-8}$\\
Later decisions with different policies & $960/960$\\
Median paired travel contrast, smaller camp & $-0.1692$\\
Median paired travel contrast, larger camp & $-0.02645$\\
\midrule
\multicolumn{2}{l}{\emph{Balanced Plurality with shared shocks, nearest support versus fixed slate}}\\
Median paired terminal voter variance contrast & $-0.000119$\\
Positive terminal voter variance contrasts & $26/64$\\
\bottomrule
\end{tabular}
\end{table}

\section{Approximating the policy objective on a grid}
\label{app:grid}
We compare the best objective value on a grid with the best value
over the full policy space at the same electorate. Suppose the policy
set is a compact convex $\Omega$ of diameter $M$. As the policy varies,
hold the common anchor $m$, coefficient $a$, and response function
$\gamma$ in \eqref{eq:anchored} fixed. Let
$|\gamma(d)|\le G$ and let $\gamma$ have Lipschitz constant
$L_\gamma$ on $[0,M]$. Write $D_+(w)$ for the zero-shock
variance after projection, and let $\delta$ be the covering
radius of the finite grid in $\Omega$.

\begin{proposition}[Error in the grid objective value]
The minimum of $R$ over the grid exceeds its minimum over
$\Omega$ by at most $\delta$. The minimum of $\sqrt{D_+}$
over the grid exceeds its minimum over $\Omega$ by at most
$(G+L_\gamma M)\delta$.
\end{proposition}

\begin{proof}
Each distance to a voter is 1-Lipschitz in $w$, as is their
maximum. In \eqref{eq:anchored}, compare the noiseless updates
at two policies $w$ and $v$. The anchor term is the same.
The difference in the policy term is bounded by
$G\norm{w-v}+L_\gamma M\norm{w-v}$, using
$|\norm{x_i-w}-\norm{x_i-v}|\le\norm{w-v}$.
Projection preserves this bound. Subtracting the mean of the
projected voter vectors is an orthogonal projection under the
normalized product norm. Its operator norm is at most one.
Thus the standard deviation $w\mapsto\sqrt{D_+(w)}$ is
$(G+L_\gamma M)$-Lipschitz. Apply each inequality from a
continuous minimizer to a grid point within $\delta$.
\end{proof}

On the unit square, the $17$ by $17$ grid has
$\delta=\sqrt2/32$. For winner-only attraction the bound
on the standard-deviation error is $0.39\delta$. For Sorting
pressure it is $(0.10+\sqrt2/8)\delta$. These are one-state
objective-value bounds. A grid can satisfy them while selecting
different policies and generating different later electorates.

\section{Gaussian supporter alignment}
\label{app:alignment}
Fix the electorate and slate. To compare support centers at different
platforms, we extend the support centroid to a function of the platform
position. For bandwidth $\sigma>0$, write
$Z_i=\sum_j\exp(-\norm{x_i-c_j}^2/\sigma^2)$.
For any $c\in\R^p$, define $s(c)$ to be the mean of the empirical
distribution on the $x_i$ with weights proportional to
$\exp(-\norm{x_i-c}^2/\sigma^2)/Z_i$. We hold the normalizers
$Z_i$ fixed in this extension, so $s(c_j)=s_j$ for each candidate
in the slate.
Cancelling the common $\exp(-\norm{c}^2/\sigma^2)$ factor shows
that the weights form an exponential tilt by $2c/\sigma^2$.
Differentiation gives
$\nabla s(c)=2\operatorname{Cov}_{c}(X)/\sigma^2$.
In any unit direction the projected voters lie in an interval of
length at most $\operatorname{diam}(\Omega)$. Squared distance to
the midpoint bounds their variance by
$\operatorname{diam}(\Omega)^2/4$. Integration along the line
segment between $c_j$ and $c_\ell$ gives the claimed constant.

\clearpage
\section{Repeated policies between two point-mass camps}
\label{app:repeated-camps}
Proposition~\ref{prop:twocamp} compares policies for a single update.
We use the same fixed response and noiseless point-mass camps to
compare sequences of policies. In each round the feasible policies
are those in the interval between the current camp positions. Let
$r_t$ be their gap after $t$ updates and $F(r)=r-\varphi(r)$ for
$0\le r\le r_0=2$. A policy at distance $u$ from the left camp
leaves a gap of $r-\varphi(u)-\varphi(r-u)$. This gap is positive
for $r>0$ because both response coefficients are at most
$\eta_{\max}<1$.

\begin{corollary}[Terminal disagreement between two camps]
At any horizon $T$, choosing either endpoint in every round
minimizes terminal disagreement over the feasible policy sequences.
All endpoint sequences give the same sequence of gaps. Each
endpoint choice assigns the whole decrease in the gap in that
round to the opposite camp.
\end{corollary}

\begin{proof}
We first compare policies at a fixed gap. Since $\varphi$ is convex and
$\varphi(0)=0$, $\varphi(u)+\varphi(r-u)\le\varphi(r)$. The next gap is
therefore at least $F(r)$, with equality at either endpoint. To compare
paths, write $r_0^*=r_0$ and $r_{t+1}^*=F(r_t^*)$ for the gaps on an endpoint
path. Since $0\le F(r)\le r$, this sequence is nonincreasing.
Suppose $r_t\in[r_t^*,r_0]$ on any feasible path. Concavity of
$F$ gives
\[
 r_{t+1}\ge F(r_t)\ge \min\{F(r_t^*),F(r_0)\}=r_{t+1}^* .
\]
The equality uses $F(r_0)=r_1^*$ and the fact that the endpoint gaps
are nonincreasing. The upper bound $r_{t+1}\le r_0$ follows because each camp moves
toward a policy in the current interval. Induction gives the gap
comparison at every horizon. Disagreement is $\pi(1-\pi)r_t^2$,
so it is minimized by an endpoint path. The camp at the selected
endpoint stays in place, and the opposite camp moves by the
decrease in the gap.
\end{proof}

For example, always selecting the left camp's position leaves that
camp fixed and assigns travel $r_0-r_T^*$ per voter to the right
camp. Always selecting the right camp's position reverses the
allocation. Selecting either endpoint with probability $1/2$
in each round gives both camps expected travel $(r_0-r_T^*)/2$
per voter, with the same realized disagreement in every case.

\end{document}